\documentclass[letterpaper, 10 pt, conference]{ieeetran}  

\IEEEoverridecommandlockouts                              
\title{A General Framework for the Optimization of Energy Harvesting Communication
Systems\\ with Battery Imperfections}

\author{Bertrand Devillers, Deniz G\"{u}nd\"{u}z \thanks{This work is supported in part by EXALTED project (IT-258512) funded by
European Union's Seventh Framework Programme (FP7), and by the Spanish
Government under project TEC2010-17816 (JUNTOS). Deniz G\"{u}nd\"{u}z is supported by the European Commission's Marie Curie IRG Fellowship with reference number 256410 under the Seventh Framework Programme.}%
\thanks{B. Devillers and D. G\"{u}nd\"{u}z are with the Centre Tecnol\`ogic de Telecomunicacions de Catalunya (CTTC), Castelldefels, Barcelona,
Spain.}
\thanks{Emails: bertrand.devillers@cttc.es, deniz.gunduz@cttc.es }}

\date{}

\usepackage{amsfonts}
\usepackage{amssymb}
\usepackage{amsmath}
\usepackage{dsfont}
\usepackage{amscd}
\usepackage{graphicx, subfigure}
\usepackage{psfrag}
\usepackage{xcolor}
\usepackage{courier}

\usepackage{cite}

\newtheorem{thm}{Theorem}[section]

\newtheorem{lem}[thm]{Lemma}

\newtheorem{defn}{Definition}[section]
\newtheorem{rem}{Remark}[section]
\newtheorem{algorithm}{Algorithm}[section]

\begin{document}
\maketitle \thispagestyle{empty} \pagestyle{empty}

\begin{abstract}
Energy harvesting has emerged as a powerful technology for complementing
current battery-powered communication systems in order to extend their
lifetime. In this paper a general framework is introduced for the optimization
of communication systems in which the transmitter is able to harvest energy
from its environment. Assuming that the energy arrival process is known
non-causally at the transmitter, the structure of the optimal transmission
scheme, which maximizes the amount of transmitted data by a given deadline, is
identified. Our framework includes models with continuous energy arrival as
well as battery constraints. A battery that suffers from energy leakage is
studied further, and the optimal transmission scheme is characterized for a
constant leakage rate.
 \end{abstract}



\section{Introduction}
\label{sec:intro}

Energy efficiency is a key challenge in the sustainable deployment of
battery-powered communication systems. Applications such as wireless sensor
networks depend critically on the lifetime of individual sensors, whose
batteries are limited due to physical constraints as well as cost
considerations. Power management is essential in optimizing the energy
efficiency of these systems in order to get the most out of the available
limited energy  in the battery. A complementary approach has recently been made
possible by introducing rechargeable batteries that can harvest energy from the
environment. Several different technologies have been proposed and implemented
for harvesting ambient energy such as solar, radio-frequency, thermoelectric or
solar (see \cite{Park:SECON:06, Le:SSC:08, Roundy:SMS:04, Alippi:TCS:08,
Paradiso:PC:05} and references therein for various examples of energy
harvesting technology).

Harvesting energy from the environment is an important alternative to
battery-run devices to extend their lifetime. However, it is important to
design the system operation based on the energy harvesting process to increase
the efficiency. Energy harvesting systems have received a lot of recent
attention \cite{Liu:INFOCOM:10, Gatzianas:WC:10, Lin:TN:07, Sharma:WC:10,
Castiglione:WIOPT:11}. Node and system level optimization have been considered
from both practical and theoretical perspectives. The previous work that are
most relevant to the problems studied in this paper are \cite{Yang:TC:10,
Tutuncuoglu:WC:10, Yang:TWC:10, Antepli:JSAC:11}. In \cite{Yang:TC:10}, the
problem of transmission time minimization is studied when the data and the
energy arrives at the transmitter in packets; and the transmission power is
optimized when the data and energy arrival times and amounts are known in
advance. In \cite{Tutuncuoglu:WC:10}, the amount of transmitted data is
maximized for an energy harvesting system under deadline and finite battery
capacity constraints. Reference \cite{Tutuncuoglu:WC:10} also shows that the
transmission time minimization problem studied in \cite{Yang:TC:10} and the
transmitted data maximization problem are duals of each other and their
solutions are identical for the same parameters. The problem is extended to the
broadcast channel in \cite{Yang:TWC:10,
Antepli:JSAC:11,Ozel:WIOPT:11,Ozel:TWC:11}, to the relay channel in
\cite{Gunduz:CAMSAP:11}, and to the multiple access channel in
\cite{Yang:ICC:11}.

The problem considered in this work is that of maximizing the amount of data
that is transmitted within a given deadline constraint under various
assumptions regarding the energy harvesting model as well as the battery
limitations. Our focus is on the offline optimization of the energy harvesting
communication system, that is, we assume that the energy arrival profile is
known in advance. We first introduce a general framework for transmitted data
maximization by adjusting the transmit power in an energy harvesting system
with battery limitations. Our model includes continuous energy harvesting,
generalizing the packetized energy arrival model considered in
\cite{Yang:TC:10} and \cite{Tutuncuoglu:WC:10}. Moreover, different from the
previous work, our model also includes the realistic scenario of battery
degradation over time by considering a time-varying battery capacity. We show
that these constraints can be modeled through cumulative harvested energy and
minimum energy curves, which are then used to obtain the optimal transmission
policy. The framework introduced for the energy harvesting system optimization
is similar to the calculus approach introduced by Zafer and Modiano for
energy-efficient data transmission in \cite{Zafer:TN:09}. We later show that
the proposed framework also applies to a broadcast channel with an energy
harvesting transmitter.

We then consider a more realistic battery model with energy leakage. Assuming a
constant leakage rate, we identify the optimal transmission strategy for the
case of a packetized energy arrival model.

The paper is organized as follows. Section \ref{s:general_model} presents the
system model. Optimal transmission scheme for a point-to-point system under
battery size constraints is derived in Sections~\ref{s:battery_size}. In
Section \ref{s:broadcast}, it is shown that the proposed framework can be used
to characterize the optimal transmission scheme in an energy-harvesting
broadcast channel. We consider battery leakage in Section
\ref{s:battery_leakage} and find the optimal transmission scheme for a linear
leakage rate.  Finally, conclusions are provided in Section
\ref{Section:conclusions}.


\section{System Model}\label{s:general_model}

We consider a continuous-time model for both the harvested  and the transmitted
energy, that is, the harvested energy is modeled as a continuous-time process,
while the transmitter is assumed to be able to adjust its transmission power,
and hence, the transmission rate, instantaneously. This continuous-time model
generalizes the discrete-time arrival model considered in \cite{Yang:TC:10} and
\cite{Tutuncuoglu:WC:10}. A cumulative curve approach is used to described the
flow of energy in the system.


\begin{defn} [Harvested Energy Curve]
The harvested energy curve $H(t)$ is a right-continuous function of time $t$,
$t\in \mathbb{R}^{+}$, that denotes the amount of energy that has been
harvested in the interval $[0,t]$.
\end{defn}

\begin{defn} [Transmitted Energy Curve]
The transmitted energy curve $E(t)$ is a continuous function with bounded right
derivative, that denotes the amount of energy that has been used for data
transmission in the interval $[0,t]$, $t\in \mathbb{R}^{+}$.
\end{defn}

Naturally, we require $E(t) \leq H(t)$, i.e., the transmitter cannot use more
energy than that has arrived. We also consider a ``minimum energy curve'' that
might model, for example, a battery size constraint.

\begin{defn} [Minimum Energy Curve]
Given an harvested energy curve $H(t)$, a minimum energy curve $M(t)$ is a
function satisfying $M(t) \leq H(t)$, $\forall t \geq 0$, and denotes the
minimum amount of energy that needs to be used by the system until time $t$.
\end{defn}

Given the harvested energy curve and the minimum energy curve, a feasible
transmitted energy curve should satisfy the conditions $M(t) \leq E(t) \leq
H(t)$, $\forall t\geq 0$. Among all feasible transmitted energy curves, our
goal is to characterize the one that transmits the highest amount of data over
a given finite time interval $[0,T]$. We consider offline optimization, that
is, the harvested and the minimum energy curves are assumed to be known in
advance\footnote{This is an accurate assumption for systems in which the energy
harvesting process can be modeled as a deterministic process. For example, in
solar based systems the amount of energy that can be harvested at various times
of the day can be modeled quite accurately. In some other systems, harvested
energy depends on the operating schedule of the harvesting device rather than
the energy source, such as shoe-mounted piezoelectric devices; and the
harvested energy curve can be modeled accurately in advance.}.

We assume that the instantaneous transmission rate relates to the power of
transmission at time $t$ through a rate function $r(\cdot)$, which is a
non-negative strictly concave increasing function of the power with $r(0) =0$.
We note here that many common transmission models, such as the capacity of an
additive white Gaussian noise channel, satisfy these conditions
\cite{Zafer:TN:09}. The total transmitted data corresponding to a given curve
$E(t)$ over the interval $[0,T]$ is found by
\begin{eqnarray}
    \mathcal{D}(E(t)) \triangleq \int_{0}^T r(E'(t)) dt,
\end{eqnarray}
where $E'(t)$ is the derivative of function $E(t)$ at time $t$, and it gives
the power of transmission at that instant while $r(E'(t))$ is the corresponding
transmission rate.

\begin{figure*}
\centering \small \subfigure[System with packet arrivals and a time-varying
battery constraint.]
    {%
        \psfrag{time}{time}\psfrag{e}{energy}
        \psfrag{Mt}{$M(t)$}\psfrag{Ht}{$H(t)$}
        \psfrag{Bt}{$b(t)$}
        \label{f:battery_time}
        \includegraphics[width=3in]{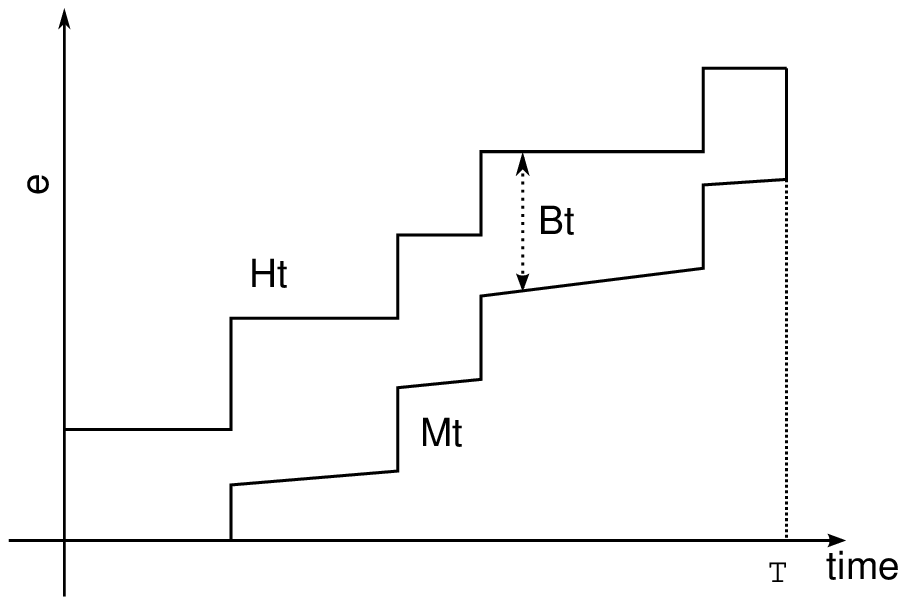}
    }%
\subfigure[The ``dying'' battery example.]
    {%
        \psfrag{t1}{$t_1$}\psfrag{t2}{$t_2$}\psfrag{t3}{$t_3$}\psfrag{tN}{$t_N$}
        \psfrag{e1}{$b_1$}\psfrag{e12}{$\hspace{-0.1in}b_1 + b_2$}\psfrag{EN}{$\hspace{-0.1in}E_N$}
        \psfrag{Mt}{$M(t)$}\psfrag{Ht}{$H(t)$}\psfrag{Eopt}{$\hspace{-0.07in}E_{opt}(t)$}
        \label{f:battery_dying}
        \includegraphics[width=3in]{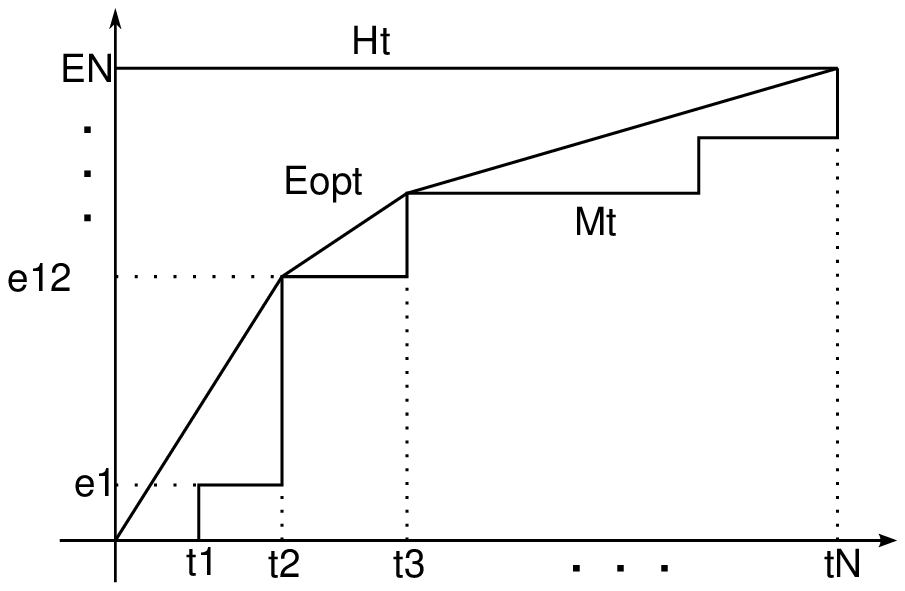}
    }%
    \caption{Illustration of the harvested and minimum energy curves for different examples.}
    \vspace{-0.5cm}
\end{figure*}

\section{Optimal Transmission Scheme under Battery Size Constraints}\label{s:battery_size}


In our problem formulation we assume that the transmitter always has data to
transmit. Hence, the minimum energy curve can be used to model a constraint on
the battery size, forcing the system to use any energy that cannot be stored in
the battery for transmission of additional data before it is discarded. For a
fixed energy curve $E(t)$ and unlimited battery size, the energy that is
available in the battery at time instant $t$ is given by $H(t)-E(t)$. However,
if the battery size is $b$, we should have $H(t)-E(t)\leq b$. Consequently, the
associated minimum energy curve is given by $M(t) = \max\{H(t)-b, 0\}$.

We can also consider a time-varying battery capacity $b(t)$, which can model
the degradation in the battery capacity over time. This is a common phenomenon
in rechargeable batteries used for energy harvesting applications. See Fig.
\ref{f:battery_time} for an illustration of the harvested and minimum energy
curves for a battery with continuously decreasing capacity.

Now, the optimization problem can be stated as follows.
\begin{eqnarray}
    \max_{E(t) \in \Gamma} & \mathcal{D}(E(t)) = \int_{0}^T r(E'(t)) dt \label{opt_single_user} \\
    \mbox{such that} & M(t) \leq E(t) \leq H(t), \forall t \in [0,T],
\end{eqnarray}
where $\Gamma$ specifies the set of all non-decreasing, continuous functions
with bounded right derivatives for all $t \in[0,T]$ and with $E(0) = 0$.




%

We first present the optimality conditions for the transmitted energy curve.
Similar to previous studies, such as \cite{Yang:TC:10},
\cite{Tutuncuoglu:WC:10}, \cite{Zafer:TN:09} and \cite{Uysal:TN:02}, our main
tool is the Jensen's inequality given in the following lemma (in the integral
form).

\begin{lem}\label{l:jensen}[Jensen's inequality]
Let $f:[a,b] \rightarrow \mathds{R}$ be a non-negative real valued function,
and $\phi(\cdot)$ be a concave function on the real line, then
\begin{align}
  \phi\left( \int_a^b f(t)dt \right) & \geq \int_a^b \frac{\phi((b-a)f(t))}{b-a}dt,
\end{align}
with strict inequality if $\phi(\cdot)$ is strictly concave, $a \neq b$, and
$f$ is not constant over the interval $[a,b]$.
\end{lem}

Consider the simple setup in which the battery has available energy $E_0$ at
time $t=0$, no further energy is harvested, and the minimum energy curve is
given as $M(t) = 0$ for $0 \leq t < T$ and $M(T)= E_0$. We will prove for this
simple setting that the constant power curve transmits the maximum amount of
data over the time interval $[0,T]$.

For any transmitted energy curve $E(t)$ with non-constant power, by replacing
the function $f$ in Lemma \ref{l:jensen} with $E'(t)/T$, and letting $a=0$,
$b=T$ and $\phi(\cdot) = r(\cdot)$, we obtain
\begin{align}
  r \left( \int_0^T \frac{E'(t)}{T} dt \right) & > \int_0^T \frac{r(E'(t))}{T}dt,
\end{align}
which is equivalent to
\begin{align}
  \int_0^T r(E'(t)) dt < Tr\left(\frac{E_0}{T}\right).
\end{align}
Note that $Tr\left(\frac{E_0}{T}\right)$ is the transmitted data by the
constant power scheme. Hence, this proves the fact that the maximum data is
transmitted by this scheme. We can express this result in a more general
context as in the following theorem.

\begin{thm}\label{t:opt_cond}
Let $E(t)$ be a feasible transmitted energy curve and $S(t)$ be a straight line
segment over interval $[a,b]$ that joins $E(a)$ and $E(b)$, $0\leq a< b\leq T$.
If $S(t)$ satisfies $M(t) \leq S(t) \leq H(t)$ for $a\leq t\leq b$, the
transmitted energy curve defined as
\begin{align}
  \hat{E}(t) & =
        \begin{cases}
            E(t), & t \in [0, a) \\
            S(t), & t \in [a, b) \\
            E(t), & t \in [b, T]
        \end{cases}
\end{align}
satisfies $\mathcal{D}(\hat{E}(t)) \geq \mathcal{D}(E(t))$.
\end{thm}

The following theorems state, respectively, the uniqueness of the optimal
transmitted energy curve and the optimality conditions. Their proofs follow
similarly to those of Theorem~2 and Lemmas~2-4 in \cite{Zafer:TN:09}.

\begin{thm}\label{t:unique}
For a strictly concave rate function $r(\cdot)$, if $\tilde{E}(t)$ is a
feasible transmitted energy curve which does not have any two points that can
be joined by a distinct feasible straight line, then $\tilde{E}(t)$ is unique
and it maximizes the transmitted data.
\end{thm}

\begin{thm}\label{t:power_change}
Let $E_{opt}(t)$ be the optimal energy expenditure curve and $t_0$ be any point
at which the power of transmission changes, i.e., the slope of $E_{opt}(t)$
changes. Then, at $t_0$, $E_{opt}(t)$ intersects either $H(t)$ or $M(t)$. If
$E_{opt}(t_0) = H(t_0)$, then the slope change must be positive. If
$E_{opt}(t_0) = M(t_0)$, then the slope change must be negative.
\end{thm}

The optimal transmitted energy curve is also the one that has the minimum
length, and hence, the same ``string visualization'' suggested in
\cite{Zafer:TN:09} can be applied here. The string visualization suggests that,
if we tie one end of a string to the origin and connect it to the point $(T,
H(T))$ tightly while constraining it to lie between $H(t)$ and $M(t)$, this
string gives us the optimal energy expenditure policy.

A special case of the framework considered here is the one with packetized
energy arrivals and without any battery constraint. This is the
energy-harvesting dual of the packet arrival problem considered in
\cite{Uysal:TN:02}. As it is shown in \cite{Tutuncuoglu:WC:10}, this is
equivalent to the problem of transmission time minimization problem studied in
\cite{Yang:TC:10}. In this problem we have $M(t) = 0$ for $t \in [0,T]$, and
$N$ energy packets arrive at times $\{t_i\}_{i=1}^{N}$. The algorithm that
gives the optimal transmitted energy curve for this problem can be obtained
following \cite{Yang:TC:10} and \cite{Zafer:TN:09}.

\begin{figure}
    \centering \small
    \psfrag{A}{$A$}
    \includegraphics[width=3in]{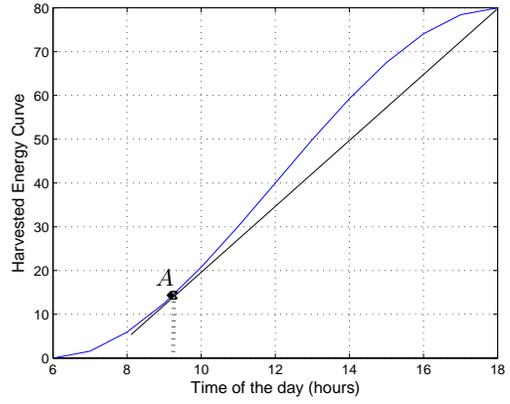}
    \caption{Continuous energy harvesting curve for a solar panel.}
        \label{f:solar_energy}
\end{figure}
%

Another example that fits into the general structure introduced above is the
following. Consider a wireless system with an energy storage unit consisting of
$N$ batteries. Assume that all the batteries are full initially and a total of
$E_N = \sum_{i=1}^N b_i$ energy is available in the system at time $t=0$, where
$b_i$ is the capacity of battery $i$. It is assumed that the batteries in the
system have finite lifetime, and they die at certain time instants, $t_i$,
$i=1, \ldots, N$. The problem is to find the maximum amount of data that can be
transmitted until the last battery dies, i.e., until $t_N$. In this problem we
have $H(t)=E_N$ for $t \in [0,t_N]$, and $M(t)$ can be obtained as in Fig.
\ref{f:battery_dying}. Note that, since once the battery dies, the energy
stored in it is not available for transmission anymore, and since we always
have data in the queue to be transmitted, it is always beneficial to use the
available energy in a battery before it dies. In this sense, we can consider
the time until a battery dies as a deadline constraint on the time the
available energy in this battery should be used. The optimal transmitted energy
curve can then be found using the string argument as seen in
Fig.~\ref{f:battery_dying}.

As an example of continuous energy arrival, we consider here a model of a solar
panel harvesting energy during the day. The amount of energy harvested per unit
of time changes during the day. While no energy is harvested when there is no
sun, the harvested energy is maximized at noon (see \cite{Scansen}). We model
the rate of harvested energy with the function $h(t)=5 - \frac{5}{36}(t-12)^2$
for $6\leq t\leq 18$, and $h(t)=0$ elsewhere, where $t \in [0,24]$ denotes the
time of the day (hours), such that $H(t)=\int_0^t h(\tau) d\tau$. The unit of
energy depends on the solar panel characteristics. The corresponding harvested
energy curve is depicted in Fig. \ref{f:solar_energy}.

Assume that we want to maximize the amount of data that can be transmitted up
to time $t=18$, i.e., until the panel stops harvesting energy. Based on the
above arguments, the optimal transmitted energy curve is identified as follows.
First we draw a tangent to the harvested energy curve from the point $(18,
H(18))$, and denote its intersection with the curve by $A$. The transmitted
energy curve follows the harvested energy curve from the origin up to $A$
\footnote{In practice, a continuous adaptation of the transmission rate is
unrealistic due to the block structure of channel coding, and the finite number
of modulation and coding modes available. However, such practical constraints
are out of the scope of this paper.}. Afterwards, it follows the straight
tangent line, i.e., it uses constant power transmission. Note that, while it is
easy to prove
 the optimality of this strategy using Theorem~\ref{t:unique}, the discrete energy
 arrival models studied in \cite{Yang:TC:10} and \cite{Tutuncuoglu:WC:10} do not apply here.

\section{Optimal Broadcast Scheme with Battery Constraint}\label{s:broadcast}

In this section, we show that the general approach introduced in Section
\ref{s:battery_size} can be instrumental in identifying the optimal
transmission policy in a broadcast channel (BC) with an energy harvesting
transmitter \cite{Yang:TWC:10}, \cite{Antepli:JSAC:11}. Consider the same
energy harvesting model at the transmitter as before; however, now there are
two receivers in the system, and the transmitter has independent data for each
receiver.


The BC problem is studied in \cite{Yang:TWC:10} and \cite{Antepli:JSAC:11};
however, the solutions in these papers are elaborated from the basics
rederiving the behavior of the optimal transmission policy in the BC scenario.
Here, we show that the general approach introduced in previous section for the
point-to-point setting can be directly applied to the BC scenario as well. This
approach allows to generalize the results in \cite{Yang:TWC:10} and
\cite{Antepli:JSAC:11} to continuous energy arrivals, and introduce battery
constraints in the problem formulation \cite{Ozel:WIOPT:11,Ozel:TWC:11}.

We consider an additive white Gaussian BC in which the signal received at
receiver $i$ is given by
\begin{align}
    Y_i = X + Z_i, ~~~~i=1,2,
\end{align}
where $X$ is the channel input of transmitter and $Z_i$ is the zero-mean
Gaussian noise component with variance $N_i$. Without loss of generality, we
assume that $N_2 > N_1 > 0$ \footnote{The case with $N_1=N_2$ reduces to the
single receiver problem.}.
Let $B_i(t)$ denote the total number of bits transmitted to receiver $i$ up to
time $t$. Our goal is to maximize the weighted sum of transmitted bits by time
$T$, $\mu_1 B_1(T) + \mu_2 B_2(T)$ for some $\mu_1, \mu_2 \geq 0$.


In the broadcast channel setting, the transmitter not only needs to identify
the transmitted energy curve $E(t)$, but also has to decide how to allocate the
power among the two receivers at each time instant. Accordingly, we denote by
$p_1(t)$ and $p_2(t)$ the power allocated to each receiver at time $t$. The
optimization problem can be written as follows.
\begin{eqnarray}\label{opt_broadcast}
    \max_{\substack{p_1(t),p_2(t)\geq 0}} &  \mu_1 \int_{0}^T r_1(t) dt + \mu_2 \int_{0}^T r_2(t) dt \\
    \mbox{such that} & M(t) \leq \int_0^t p_1(\tau) + p_2(\tau) d\tau \leq H(t), t \in [0,T] \nonumber
\end{eqnarray}
We assume that the rate-power functions are operating on the boundary of the
capacity region of the Gaussian BC:
\begin{align}
    r_1(t) &= \frac12 \log_2 \left(1+\frac{p_1(t)}{N_1}\right) \\
    r_2(t) &= \frac12 \log_2 \left(1+ \frac{p_2(t)}{p_1(t)+N_2} \right).
\end{align}
The considered optimization can be decoupled into two maximization problems as
follows:
\begin{equation}\small
    \max_{\substack{ E(t)\in \Gamma \\ M(t) \leq E(t) \leq H(t) }} \int_0^T \left[ \max_{\substack{p_1(t),p_2(t)\geq 0 \\ p_1(t) + p_2(t)=p(t)}} \mu_1 r_1(t)  + \mu_2 r_2(t) \right]  dt, \label{eq: problem decoupled}
   \end{equation}
where we define $p(t) = E'(t)$.


First, we consider the maximization problem in between brackets in \eqref{eq:
problem decoupled}. Defining $\mu \triangleq \frac{\mu_2}{\mu_1}$, we can make
the following observations on its solution\footnote{The time variable $t$ is
omitted for conciseness.}:
\begin{itemize}
  \item If $\mu > \frac{N_2}{N_1}$, no power is allocated to the first receiver, i.e. $p_1=0$, independent of the total power.
  \item If $\mu \leq 1$, no power is allocated to the second receiver, i.e. $p_2=0$, independent of the total power.
  \item When $1 <\mu \leq \frac{N_2}{N_1}$, the optimal power allocation behaves as follows. If the available total power is below $p_{th} \triangleq \frac{N_2-\mu N_1}{\mu-1}$, all the total power is allocated to receiver 1, i.e., $p_1=p$ and $p_2=0$. On the other hand, if $p \geq p_{th}$, then we have $p_1 = p_{th}$ and $p_2 = p-p_{th}$.
\end{itemize}
Note that, if $\mu > \frac{N_2}{N_1}$ or $\mu \leq 1$, the problem reduces to
the point-to-point setting; hence, we assume $1 < \mu \leq \frac{N_2}{N_1}$ in
the remainder. We can write the outcome of the maximization problem in between
brackets in \eqref{eq: problem decoupled} as
\begin{align}\label{rate_broadcast}\small{
    r(p) \triangleq
        \begin{cases}
            \frac{\mu_1}{2} \log_2 \left(1+\frac{p}{N_1}\right) \hfill \text{if } 0\leq p \leq p_{th},\\
            \frac{\mu_1}{2} \log_2 \left(1+\frac{p_{th}}{N_1}\right) + \frac{\mu_2}{2} \log_2 \left(1+ \frac{p-p_{th}}{p_{th}+N_2}\right)  \text{if } p_{th} \leq p.
        \end{cases}}
\end{align}

Then we can rewrite the optimization problem in \eqref{eq: problem decoupled}
in the same form as the point-to-point problem in (\ref{opt_single_user}) with
a rate function given in \eqref{rate_broadcast}. We next prove that this rate
function is strictly concave.

\begin{lem}
The rate function $r(p)$ in (\ref{rate_broadcast}) is a strictly concave
function of power $p$.
\end{lem}
\begin{proof}
It is easy to show that $r(p)$ is continuous, differentiable, and its
derivative is decreasing with~$p$; hence, it is a strictly concave function of
$p$.
\end{proof}

Now, based on this form of the optimization problem, we can directly use the
results of Section \ref{s:battery_size} in the broadcast channel setting in
order to identify the optimal transmission scheme for an energy harvesting
transmitter. Note that as opposed to \cite{Yang:TWC:10} and
\cite{Antepli:JSAC:11}, our solution is valid for continuous energy arrivals as
well as transmitters with various battery constraints. Once the optimal total
transmit power over time is characterized, the power allocation among the users
at each instant can be found using (\ref{rate_broadcast}).

\section{Optimal Transmission Scheme with Battery Leakage}\label{s:battery_leakage}

In Sections \ref{s:battery_size} and \ref{s:broadcast} and references therein,
the battery has been considered to be ideal, that is, there was no energy
leakage. In this section, we consider the more realistic scenario of a battery
that leaks part of the stored energy.


The leakage rate of a battery depends on the type (Li-ion batteries have a
lower leakage rate compared to the nickel-based ones), age and usage of the
battery as well as the medium temperature. Moreover, even for a fixed type of
battery and medium temperature, the leakage rate changes over time; the
batteries leak most right after being charged. However, for simplicity, a
constant rate leakage model is considered here. If the battery is non-empty at
a given time instant, the energy is assumed to leak from the battery at a
constant finite rate denoted by $\epsilon\geq 0$. Obviously no leakage occurs
if the battery is empty.  We use the same cumulative curve approach to model
the battery leakage process. Note that the leakage rate $\epsilon$ can
alternatively be interpreted as the constant operation power of the node, that
is, the circuit power needed to maintain the node awake.


\begin{defn} [Energy Leakage Curve] The energy leakage curve $L(t)$ is the amount of energy that has leaked from the battery in the
time interval $[0,t]$, $t\in \mathbb{R}^{+}$, with $L(0)=0$. Due to the
constant leakage rate assumption, $L(t)$ is a continuous, non-decreasing
function whose right-derivative is given by
\begin{align}
L'_{+}(t) =
       \begin{cases}
        \epsilon, & \text{if}\quad E(t)<H(t)-L(t),\\
        0, & \text{otherwise}.
       \end{cases}
\end{align}
\end{defn}

To highlight the effect of leakage, we do not consider any minimum energy curve
in this section, i.e., $M(t)=0$ $\forall t$, and we focus only on discrete
energy packet arrivals. Defining a maximum energy curve as $U(t) \triangleq
H(t)-L(t)$, the feasibility condition on the transmitted energy curve becomes
$0\leq E(t)\leq U(t)$. We tackle again the problem of characterizing the
feasible transmitted energy curve that transmits
 the most data over a given finite time interval $[0,T]$. The corresponding optimization problem can be stated as
\begin{eqnarray}
    \max_{E(t) \in \Gamma} & \mathcal{D}(E(t)) = \int_{0}^T r(E'(t)) dt \label{eq: problem leakage} \\
    \mbox{such that} &  0\leq E(t)\leq U(t),  \forall t \in [0,T]. \label{eq: constraint leakage}
\end{eqnarray}


\begin{rem}
Unlike the battery size constraint studied in Section \ref{s:battery_size}, the
battery leakage phenomenon does not translate into a minimum energy curve, but
into a maximum energy curve obtained by removing the total leaked energy from
the harvested energy curve. More importantly, the leakage curve is a function
of the transmitted energy curve. Consequently, the maximum energy curve
inherently depends on the transmitted energy curve, and hence, the solution
framework presented in Section \ref{s:battery_size} does not directly extend to
this setup.
\end{rem}




Throughout this section,we consider the discrete energy harvesting process in
which the $n$-th energy packet of size $E_n$ arrives at time instant $t_n$ for
$n=1, \ldots, N$. Without loss of generality, the first packet is assume to
arrive at time $t=0$ (i.e., $t_1=0$). We call this general setup the $N$-packet
problem. As before, we assume that the transmitter always has enough data to
transmit. Below, we characterize the optimal transmission scheme first for the
single-packet problem (i.e., $N=1$), and then for the general $N$-packet
problem.

\subsection{The Single-Packet Problem}
\label{s:single packet problem}

\begin{figure}
\centering \small \psfrag{time}{time}\psfrag{e}{energy}
\psfrag{Ut}{{\color{red}$U(t)$}}\psfrag{Ht}{$H(t)$}
\psfrag{Et}{{\color{blue}$E(t)$}} \psfrag{p}{{\color{blue}$p$}} \psfrag{E}{$E$}
\psfrag{eps}{{\color{red}-$\epsilon$}}\psfrag{t}{$t=\frac{E}{p+\epsilon}$}
\includegraphics[width=2.7in]{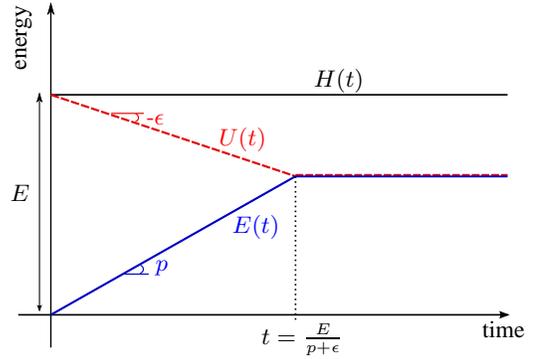}
\caption{The S$_\infty$ problem.} \label{f:SP problem}
\end{figure}

\begin{figure}
\centering \small \psfrag{time}{time}\psfrag{e}{energy}
\psfrag{Uti}{{\color{red}$\widetilde{U}(t)$}}\psfrag{Ht}{$H(t)$}
\psfrag{Et}{{\color{blue}$E(t)$}} \psfrag{p}{{\color{blue}$p$}}
\psfrag{E1}{$E$} \psfrag{eps}{{\color{red}-$\epsilon$}}\psfrag{T}{$T$}
\psfrag{s}{{\color{blue}$s=\frac{E}{T}-\epsilon$}} \psfrag{A}{$A$}
\includegraphics[width=2.7in]{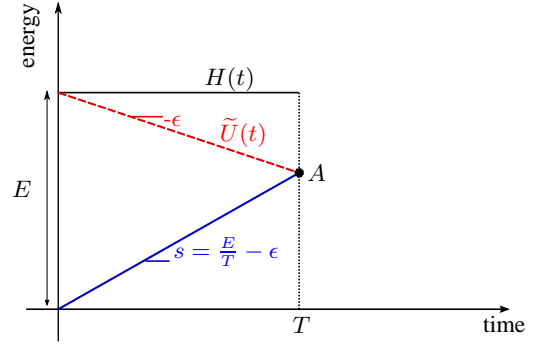}
\caption{The S$_T$ problem.} \label{f:SP problem finite}
\end{figure}

We consider here the simplified problem consisting of a single energy packet
$E$ harvested at time $t=0$. We refer to it as the single-packet problem. The
solution of this problem will serve as a building block for the general
$N$-packet problem.

First, let us treat the single-packet problem with infinite\footnote{Note that,
in the case of energy leakage, potential transmit time is finite when the
number of harvested energy packets is finite as the available energy decays to
zero even if no data is transmitted.} deadline constraint (i.e., $T=\infty$),
and denote it by S$_\infty$. It is depicted in Fig.~\ref{f:SP problem}.
Following Section \ref{s:battery_size}, it is not hard to show that the optimal
transmitted energy curve $E(t)$ has to be piecewise linear, and the slope
changes occur only if $E(t)$ intersects $U(t)$. Consequently, the optimal
$E(t)$ for the S$_\infty$ problem is as shown in Fig.~\ref{f:SP problem}: the
node transmits at a constant power $p$ until the battery runs out of energy.
One can see that there is a trade-off in the choice of $p$: while it is more
energy efficient to transmit at lower power for a longer period of time, the
longer the transmission time, the more energy will be wasted due to leakage.
The optimization problem in \eqref{eq: problem leakage}-\eqref{eq: constraint
leakage} becomes
\begin{eqnarray}
\max_{p\geq0} & \mathcal{D}(E(t)) = \frac{E}{p+\epsilon}\ r(p).
\end{eqnarray}
Assuming that $r(p)$ is a strictly concave increasing function with $r(0)=0$,
and a finite leakage rate $\epsilon$, the function \mbox{$f(p) \triangleq
\frac{r(p)}{p+\epsilon}$} achieves its maximum at a finite $p\in
\mathbb{R}^{+}$, as shown in Appendix~\ref{s:appendix properties r}. We denote
the corresponding optimal value by $p^*$. Note that while the total amount of
transmitted data is proportional to $E$,  $p^*$ is independent of~$E$.
Summarizing, the optimal transmission strategy for the S$_\infty$ problem is to
transmit at constant power $p^*$ until the battery is empty. The total amount
of transmitted data is $\frac{E}{p^*+\epsilon}\ r(p^*)$.

We next consider the single-packet problem with a fixed transmission deadline
$T$, and denote it by S$_T$. It is depicted in Fig.~\ref{f:SP problem finite},
and the following notations are defined: \mbox{$\widetilde{U}(t) \triangleq
H(t) - \epsilon t$} (we assume $\widetilde{U}(t)>0$ for all $0 \leq t \leq T$,
as otherwise the problem is equivalent to the S$_\infty$ problem). We denote
the point $(T, \widetilde{U}(T))$ by $A$. Finally, the slope of the line
segment from the origin to $A$ is denoted by $s$. We have $s = E/T - \epsilon$.
As before $p^*$ denotes the value that maximizes the function $f(p)$. Note
that, as shown in Appendix~\ref{s:appendix properties r}, $f(p)$ is strictly
decreasing for $p>p^*$. Hence, building on the solution derived for the
S$_\infty$ problem, the solution of the S$_T$ is easily derived:
\begin{itemize}
\item if $s<p^*$, transmit at constant power $p^*$ until the battery is empty.
\item else, transmit at constant power $s$ during the whole $[0,T]$ interval.
\end{itemize}
In short, the optimal transmission strategy for the S$_T$ problem is to
transmit at constant power $\tilde{p}=\max\left(p^*,s\right)$ for a time
duration $\frac{E}{\tilde{p}+\epsilon}$ (that is, until the battery is empty),
and remain silent afterwards. The amount of transmitted data is
$\frac{E}{\tilde{p}+\epsilon}r(\tilde{p})$.

\begin{figure*}[tb!]
\centering \small \psfrag{time}{time}\psfrag{e}{energy}
\psfrag{Ut}{{\color{red}$U(t)$}}\psfrag{Ht}{$H(t)$} \psfrag{Lt}{$L(t)$}
\psfrag{tUt}{{\color{red}$\tilde{U}(t)$}} \psfrag{Et}{{\color{blue}$E(t)$}}
\psfrag{p}{{\color{blue}$\tilde{p}$}} \psfrag{E}{$E=2\bar{E}$}
\psfrag{E1}{$\bar{E}$} \psfrag{E2}{$\bar{E}$} \psfrag{E3}{$E_3$}
\psfrag{a1}{$t_1$} \psfrag{a2}{$t_2$} \psfrag{a3}{$t_3$}
\psfrag{T1}{$\widetilde{T}_1$} \psfrag{T}{$T$}
 \psfrag{T2}{$\widetilde{T}_2$} \psfrag{T3}{$\widetilde{T}_3$}
\psfrag{eps}{{\color{red}-$\epsilon$}}\psfrag{t}{$\frac{E}{p^*+\epsilon}$}
\psfrag{a}{\scriptsize{(a) Original N$_T$ problem}} \psfrag{b}{\scriptsize{(b)
Original N$_T$ problem}}
\includegraphics[width=6.5in]{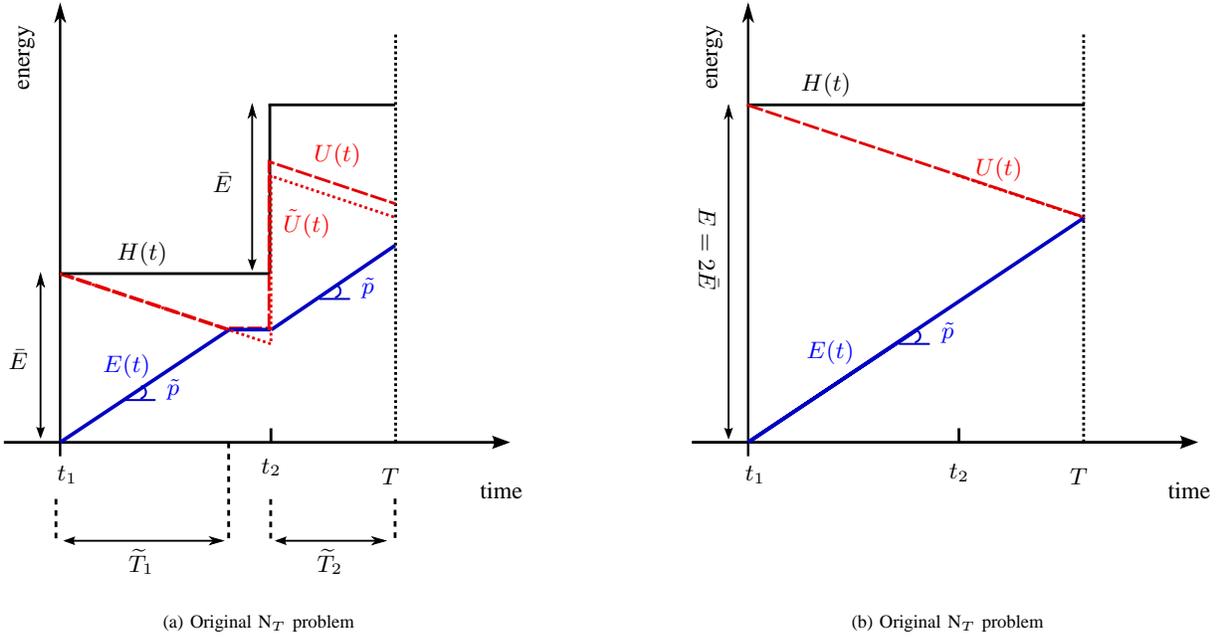}
\caption{Counterexample for the equivalence of the N$_T$ and S$_T$ problems,
with $N=2$ energy packets. } \label{f:N Packets problem counterexample}
\end{figure*}

\subsection{The $N$-Packet Problem}

We consider here the general $N$ packet problem with finite deadline constraint
$T$, denoted as the N$_T$ problem.

We start with the following lemma which proves that the optimal solution of the
N$_T$ problem can be emulated in the equivalent S$_T$ problem with
$E=\sum_{n=1}^{N}E_n$. That is, having all energy packets at time $t=0$ is at
least as good as having them arrive over time. Let us denote by $D_{N_{T}}$ and
$D_{S_{T}}$ the optimal solutions (in terms of total transmitted data) of the
N$_T$ and equivalent S$_T$ problems, respectively.

\begin{lem}
\label{Lemma1} The optimal solution of the N$_T$ problem can be obtained in the
equivalent S$_T$ problem with $E=\sum_{n=1}^{N}E_n$. That is, we have
$D_{S_{T}}\geq D_{N_{T}}$.
\end{lem}
\begin{proof}
Consider the optimal curve for the $N_T$ problem, and divide the $[0,T]$ time
interval into $N$ sub-intervals: $[t_1,t_2]$, $[t_2,t_3]$, $\dots$,
$[t_n,t_{n+1}]$, $\dots$, $[t_N,T]$. We denote by $T_n$ the duration of the
$n^{th}$ interval, i.e., $T_n\triangleq t_{n+1}-t_n$ for $n=1,\dots,N-1$, and
$T_N\triangleq T-t_N$. From Theorem~\ref{t:unique}, we know that the optimal
transmitted energy curve is a piecewise linear function, which is composed of
constant power periods possibly separated by silent intervals (i.e. horizontal
segments) in case the battery runs out of energy. Accordingly, we define the
optimal solution of the $N_T$ problem by the sequences
$\{\overline{p}_1,\overline{p}_2,\dots,\overline{p}_N\}$ and
$\{\overline{T}_1,\overline{T_2},\dots,\overline{T}_N\}$, meaning that the node
transmits for time $\overline{T}_n$ (with $\overline{T}_n\leq T_n$) at power
$\overline{p}_n>0$ in the $n^{th}$ interval. The node is silent in the
remainder of the interval, i.e., for time $T_n-\overline{T}_n$.

The data transmitted by this transmission strategy is \mbox{$D_{N_{T}}=\sum_{n=1}^N
\overline{T}_n\ r(\overline{p}_n)$}. The total transmit energy is $\sum_{n=1}^N
\overline{T}_n\ \overline{p}_n$, while the total energy leakage is $\epsilon
\sum_{n=1}^N \overline{T}_n$. Since the optimal solution should eventually
empty the battery, we have
\begin{equation}
\sum_{n=1}^N \overline{T}_n\ \overline{p}_n+\epsilon  \sum_{n=1}^N
\overline{T}_n\ =\sum_{n=1}^N E_n. \label{eq: energy conservation lemma 1}
\end{equation}

We now argue that this optimal solution can be emulated in the S$_T$ problem
with $E=\sum_{n=1}^{N}E_n$. Consider the following transmission strategy $E(t)$
for the S$_T$ problem: transmit at constant power equal to $\overline{p}_1$ for
time $\overline{T}_1$, followed by $\overline{p}_2$ for time $\overline{T}_2$,
and so on, ending with $\overline{p}_N$ for time $\overline{T}_N$. By
construction, this strategy transmits the same amount of data $D_{N_{T}}$ as
the optimal solution of the N$_T$ problem. We conclude the proof by showing
that this strategy is feasible, that is, $E(t)\leq U(t)$ for all $t\in[0,T]$.
Since the node is constantly transmitting during the interval
$[0,\sum_{n=1}^{N}\overline{T}_n]$, the curve $U(t)$ is constantly
decreasing\footnote{We assume $U(t)>0$ in the considered time interval, as
otherwise the problem can be divided into equivalent subproblems.} during this
interval at rate~$\epsilon$, i.e.~$U(t)=\sum_{n=1}^{N}E_n-\epsilon t$, for
$t\in [0,\sum_{n=1}^{N}\overline{T}_n]$. We have
\begin{eqnarray}
U\left(\sum_{n=1}^{N}\overline{T}_n\right)&=&\sum_{n=1}^{N}E_n-\epsilon \sum_{n=1}^{N}\overline{T}_n \\
&=& \sum_{n=1}^N \overline{T}_n\ \overline{p}_n \label{eq:proof feasibility} \\
&=& E\left(\sum_{n=1}^{N}\overline{T}_n\right)
\end{eqnarray}
where the equality \eqref{eq:proof feasibility} follows from \eqref{eq: energy
conservation lemma 1}. This proves the feasibility of $E(t)$.

Having proved the achievability of $D_{N_{T}}$ in the equivalent S$_T$ problem,
the inequality $D_{S_{T}}\geq D_{N_{T}}$ naturally follows.
\end{proof}


The counterpart of Lemma \ref{Lemma1} in the other direction does not always
hold, that is, the optimal solution of the equivalent S$_T$ problem cannot
always be emulated in the original N$_T$ problem. A counterexample can indeed
easily be constructed, which is depicted in Fig.~\ref{f:N Packets problem
counterexample}. Part (a) of the figure depicts a 2-packet problem with
$E_1=E_2=\bar{E}$, $t_1=0$, and $t_2>T/2$. In part (b) the equivalent S$_T$
problem is depicted.  Let the optimal transmission power for the S$_T$ problem
be given by  $\tilde{p}=\frac{2\bar{E}}{T}-\epsilon$. This solution cannot be
emulated in the original N$_T$ problem. In fact, as shown in part (a), the node
cannot transmit a constant power $\tilde{p}$ during the full $[0,T]$ time
interval as the battery runs out of energy at time $T/2$, and the node has to
remain silent during the time interval $[T/2,t_2]$.

However, in the following lemma, we provide a sufficient condition for the
counterpart of Lemma \ref{Lemma1} to hold. For this, we define $A_i$ as the
point on the $\widetilde{U}(t)$ curve corresponding to the time instant
$t_{i+1}$, for $i=1,2,\ldots,N-1$, and $A_N$ as the point corresponding to time
$t=T$, as illustrated in Fig.~\ref{f:NP problem finite}.

%
%

\begin{figure}
\centering \small \psfrag{time}{time}\psfrag{e}{energy}
\psfrag{Uti}{{\color{red}$\widetilde{U}(t)$}}\psfrag{Ht}{$H(t)$}
\psfrag{Et}{{\color{blue}$E(t)$}} \psfrag{p}{{\color{blue}$p$}}
\psfrag{E1}{$E_1$} \psfrag{eps}{{\color{red}-$\epsilon$}}\psfrag{T}{$T$}
\psfrag{s0}{{\color{blue}$s_0$}}\psfrag{s1}{{\color{blue}$s_1$}}\psfrag{sn}{{\color{blue}$s_n$}}
\psfrag{A1}{$A_1$}
\psfrag{A2}{$A_2$}\psfrag{A3}{$A_3$}\psfrag{AN}{$A_N$}\psfrag{An}{$A_n$}\psfrag{E1}{$E_1$}
\psfrag{E2}{$E_2$} \psfrag{E3}{$E_3$} \psfrag{a1}{$t_1$} \psfrag{a2}{$t_2$}
\psfrag{a3}{$t_3$} \psfrag{A}{$A$}
\includegraphics[width=2.7in]{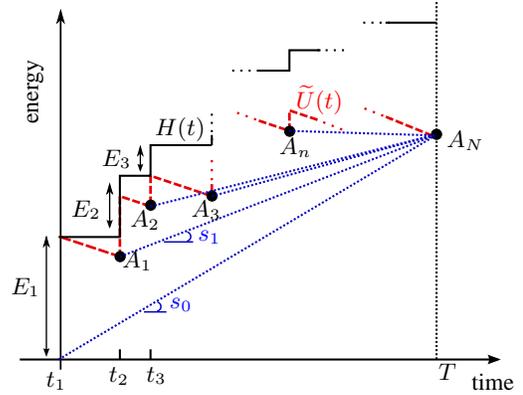}
\caption{A $N_T$ problem satisfying the conditions of Lemma~\ref{Lemma2}.}
\label{f:NP problem finite}
\end{figure}

\begin{lem}
\label{Lemma2} If the line segment from the origin to the point $A_N$ does not
cross the curve $\widetilde{U}(t)$ at any other point than $\{A_1,\dots,A_N\}$,
then the optimal solution of the S$_T$ problem with $E=\sum_{n=1}^{N}E_n$ can
be obtained in the N$_T$ problem, and $D_{N_{T}}\geq D_{S_{T}}$. This
sufficient condition is expressed by the following $N-1$ inequalities:
\begin{equation}
\frac{\sum_{n=1}^{i} E_n}{\sum_{n=1}^{i}T_n}\geq\frac{\sum_{n=1}^{N} E_n}{T},
\qquad i=1,\dots,N-1. \label{eq:inequalities suff condition}
\end{equation}
\end{lem}
\begin{proof}
First note that the set of inequalities in \eqref{eq:inequalities suff
condition} expresses that the line segments from origin to points $A_1,
A_2,\dots,A_{N-1}$ have slope which are all greater than that of the segment
from origin to $A_N$. This requires that the line segments from the points
$A_1, A_2,\dots,A_{N-1}$ to the point $A_N$ have slopes $s_1,s_2,\dots,s_{N-1}$
respectively, which all are lower than or equal to the slope
$s_0=\frac{\sum_{n=1}^{N} E_n}{T}- \epsilon$ of the segment from the origin to
point $A_N$:
\begin{equation}
s_i \leq s_0= \frac{\sum_{n=1}^{N} E_n}{T}- \epsilon \label{eq:inequalities
suff condition bis}
\end{equation}
for $i=1,\dots,N-1$. 
An illustration of an N$_T$ problem satisfying the conditions of this lemma is
given in Fig.~\ref{f:NP problem finite}.


Consider now the S$_T$ problem with $E=\sum_{n=1}^{N}E_n$. Remember that the
optimal scheme for the S$_{T}$ problem requires transmitting at constant power
$\tilde{p}$ for a duration of $\frac{E}{\tilde{p}+\epsilon}$, with
$\tilde{p}=\max\left(p^*,s_0\right)\geq s_0$. We now argue that this solution
can be emulated in the N$_T$ problem. Consider the following transmission
strategy for the N$_T$ problem: transmit at $\tilde{p}$ whenever the battery is
non-empty, and remain silent otherwise. By construction, this strategy is
feasible. Again consider the $N$ time intervals between energy arrivals
$[t_1,t_2]$, $[t_2,t_3]$, $\dots$, $[t_N,T]$ of durations $T_1,T_2,\dots,T_N$,
respectively. We denote by $\tilde{T}_n$ (with $\tilde{T}_n\leq T_n$) the time
for which the node is transmitting in the $n^{th}$ interval. The total
transmission time is then given by $T_{tot} \triangleq
\sum_{n=1}^{N}\tilde{T}_n$. Moreover, combining the inequalities in
\eqref{eq:inequalities suff condition bis} with the fact that
$\tilde{p}=\max\left(p^*,s_0\right)\geq s_0$, we have that $\tilde{p}\geq s_i$
for $i=1,\dots, N-1$. This ensures that the considered strategy uses up the
whole available energy by time $T$, i.e., $E(T)=U(T)$. Then, by the
conservation of energy, the transmit and leakage energies must sum to the total
harvested energy:
\begin{equation}
T_{tot}\ \tilde{p} + \epsilon\ T_{tot} = \sum_{n=1}^N E_n = E,
\end{equation}
from which we get that $T_{tot}=\frac{E}{\tilde{p}+\epsilon}$, just like for
the optimal solution of the S$_T$ problem. This transmission strategy thus
transmits the same amount of data $D_{S_{T}}$ as the optimal solution of the
S$_T$ problem. Consequently, under the conditions given in the theorem, the
inequality $D_{N_{T}}\geq D_{S_{T}}$ holds.
\end{proof}

Building on the two previous lemmas, the following theorem can be formulated.
\begin{thm}
\label{Theorem line of sight problem} If the inequalities in
\eqref{eq:inequalities suff condition} hold, then:
\begin{itemize}
\item[(i)] $D_{N_{T}}=D_{S_{T}}$, that is, the optimal solutions of the N$_T$ problem and the S$_T$ problem with $E=\sum_{n=1}^{N}E_n$ are equivalent.
\item[(ii)] The optimal transmission strategy for the N$_T$ problem is to transmit at constant power $\tilde{p}$ whenever the battery is non-empty, and remain
silent otherwise, where the value $\tilde{p}$ corresponds to the solution of
the equivalent S$_T$ problem:
\begin{equation}
\tilde{p}=\max\left(p^*,\frac{\sum_{n=1}^{N} E_n}{T}- \epsilon\right)
\end{equation}
The total amount of transmitted data is
$\left(\frac{\sum_{n=1}^{N}E_n}{\tilde{p}+\epsilon}\right) r(\tilde{p})$.
\end{itemize}
\end{thm}

\begin{figure*}[tb!]
\centering \small \psfrag{time}{time}\psfrag{e}{energy}
\psfrag{Ut}{{\color{red}$U(t)$}}\psfrag{Ht}{$H(t)$} \psfrag{Lt}{$L(t)$}
\psfrag{tUt}{{\color{red}$\tilde{U}(t)$}} \psfrag{Et}{{\color{blue}$E(t)$}}
\psfrag{p}{{\color{blue}$p^*$}} \psfrag{E}{$E=E_1+E_2+E_3$} \psfrag{E1}{$E_1$}
\psfrag{E2}{$E_2$} \psfrag{E3}{$E_3$} \psfrag{a1}{$t_1$} \psfrag{a2}{$t_2$}
\psfrag{a3}{$t_3$} \psfrag{T1}{$\widetilde{T}_1$} \psfrag{T}{$T$}
 \psfrag{T2}{$\widetilde{T}_2$} \psfrag{T3}{$\widetilde{T}_3$}
\psfrag{eps}{{\color{red}-$\epsilon$}}\psfrag{t}{$\frac{E}{p^*+\epsilon}$}
\psfrag{a}{\scriptsize{(a) original N$_T$ problem}} \psfrag{b}{\scriptsize{(b)
equivalent S$_T$ problem}}
\includegraphics[width=6.5in]{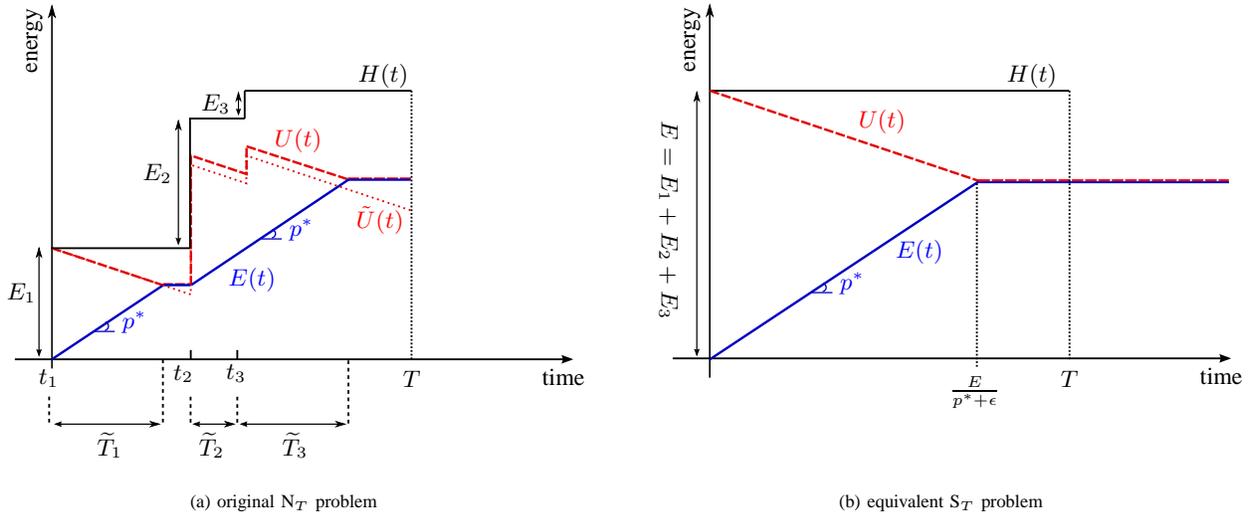}
\caption{Illustration of Theorem~\ref{Theorem line of sight problem}, with
$N=3$ energy packets and $\tilde{p}=p^*$. } \label{f:N Packets problem}
\end{figure*}

An illustration of the result in Theorem~\ref{Theorem line of sight problem} is
provided in Fig.~\ref{f:N Packets problem} for $N=3$ and $\tilde{p}=p^*>s_0$.
Part (a) of the figure depicts the N$_T$ problem, while its equivalent S$_T$
problem is given in part~(b). According to Theorem \ref{Theorem line of sight
problem}, for both problems the optimal strategy is to transmit at constant
power $p^*$. The only particularity of the N$_T$ problem is the presence of
silent zones in between energy packet arrivals. However, the distribution over
time of these silent zones do not affect the total duration of transmission,
guaranteeing the equivalence of both solutions in terms of amount of
transmitted data.

Now, building on Theorem~\ref{Theorem line of sight problem}, we can provide
the optimal solution for any N$_T$ problem. Consider all line segments
connecting the origin to points $A_i$, $i=1,\dots,N$. Among the segments that
do not intersect $\widetilde{U}(t)$ other than at point $\{A_1,\dots,A_N\}$, we
pick the one with the highest index, i.e., the rightmost end point. We denote
this index by $k$. We can now consider the $k$ first energy packets only, and
solve the corresponding $k$ packet problem with deadline $\sum_{n=1}^k T_n$,
using the equivalence given in Theorem~\ref{Theorem line of sight problem}. We
then proceed recursively by considering the remaining $N-k$ packet problem
separately. This recursive algorithm is described next. It takes as inputs the
number of packets $N$, the sizes of energy packets $\{E_n\}_{n=1}^N$, and the
packet interarrival times $\{T_n\}_{n=1}^N$. It returns as output the set of
optimal transmission powers $\{\tilde{p}_n\}_{n=1}^N$, meaning that the optimal
solution of the N$_T$ problem is to transmit at constant power $\tilde{p}_n$ in
the $n^{th}$ interval as long as the battery is non-empty.   The optimality of
the algorithm is proved in Appendix~\ref{Appendix proof algo}.

\vspace{0.5cm} \vline{\hspace{0.3cm}\begin{minipage}{3.5in}\vspace{0.0cm}
\begin{algorithm} \label{algo NT problem} \emph{N$_{T}$-problem}$(N,E_1,\hdots,E_N,T_1,\hdots,T_N)$\\
 \textbf{Input}:
\begin{itemize}
\item $N$: number of energy packets
\item $\{E_n\}_{n=1}^N$: amount of energy in each packet
\item $\{T_n\}_{n=1}^N$: interarrival times
\end{itemize}
\textbf{Output}: $\{\tilde{p}_n\}_{n=1}^N$\\
 \textbf{Algorithm}:
\begin{enumerate}
\item Find the highest $k \in \{1,\dots,N\}$ such that
\begin{equation}
\frac{\sum_{n=1}^{i} E_n}{\sum_{n=1}^{i}T_n}\geq\frac{\sum_{n=1}^{k}
E_n}{\sum_{n=1}^{k}T_n} \label{eq:ineq algo}
\end{equation}
for all  $i\in \{1,\dots,k-1\}$.

\item
\begin{equation}
\tilde{p}_i=\max\left(p^*,\frac{\sum_{n=1}^{k} E_n}{\sum_{n=1}^{k}
T_n}-\epsilon\right) \label{eq:power algo}
\end{equation}
for all  $i\in \{1,\dots,k\}$.
\item If $k<N$, find the $\{\tilde{p}_n\}_{n=k+1}^N$ by running \newline \emph{N$_{T}$-problem}$(N-k,E_{k+1},\hdots,E_N,T_{k+1},\hdots,T_N)$
\end{enumerate}

\end{algorithm}
\end{minipage}}\vspace{0.5cm}\\

We conclude this section by identifying two special cases of the solution
provided here:
\begin{itemize}
\item The special case of an $N$-packet problem without deadline
constraint can be solved by Algorithm \ref{algo NT problem} by setting
$T_N=\infty$. In this case, the inequalities in \eqref{eq:ineq algo} hold with
$k=N$, and \eqref{eq:power algo} reduces to $\tilde{p}_i=\max(p^*,0)=p^*$ for
all $i\in \{1,\dots,N\}$. Hence, the optimal transmission strategy for the
$N$-packet problem without deadline constraint is to transmit at constant power
$p^*$ whenever the battery is non-empty, and remain silent otherwise.
\item The special case of a perfect battery with no leakage is obtained by
setting $\epsilon=0$. In this case, $p^*=0$ (as detailed in Appendix
\ref{s:appendix properties r}), and Algorithm \ref{algo NT problem} reduces to
the solution proposed in~\cite{Yang:TC:10}.
\end{itemize}

\section{Conclusion}\label{Section:conclusions}

We have considered a communication system with an energy harvesting
transmitter. Taking into account various constraints on the battery we have
optimized the transmission scheme in order to maximize the amount of data
transmitted within a given transmission deadline. We have provided a general
framework extending the previous work in \cite{Yang:TC:10} and
\cite{Tutuncuoglu:WC:10} to the model with continuous energy arrival as well as
time-varying battery size constraints. We have also showed that the proposed
framework applies to the optimization of energy harvesting broadcast systems.
Moreover
we have studied the case of a battery suffering from energy leakage, for which
the optimal transmission scheme has been characterized for a constant leakage
rate.


\begin{appendix}
\subsection{Properties of $f(p) \triangleq \frac{r(p)}{p+\epsilon}$}
\label{s:appendix properties r} Remember that $r(p)$ is a non-negative strictly
concave increasing function, with $r(0)=0$.  We prove here that the function
$f(p) \triangleq \frac{r(p)}{p+\epsilon}$, with $p\geq0$, achieves its maximum
at a finite $p^* \in \mathbb{R}^{+}$, and is strictly decreasing for $p > p^*$.

 The derivative of $f(p)$ is calculated as follows:
\begin{equation}
f'(p)=\frac{r'(p)(p+\epsilon)-r(p)}{(p+\epsilon)^2} \label{eq: appendix}
\end{equation}
 We distinguish
two cases:\vspace{0.2cm}\newline
 1) If $\epsilon=0$, \eqref{eq: appendix} becomes
    \begin{equation}
    f'(p)=\frac{r'(p)p-r(p)}{p^2},
    \end{equation}
    which is analyzed as follows:
    \begin{itemize}
    \item if $p=0$, both the numerator and the denominator are zero. By l'H\^opital's
    rule, we get $\lim_{p\rightarrow0} f'(p)=r''(0)/2<0$, which
    follows from the strict concavity of $r(p)$.
    \item  if $p>0$, the numerator $r'(p)p-r(p)$ is a strictly negative function. Indeed,
    the strict concavity of $r(p)$ together with the fact that $r(0)=0$
    guarantees
    that $r(p)> r'(p)p$ for all $p > 0$.
    \end{itemize}
    Overall, $f'(p)$ is thus strictly negative for all $p\geq 0$. Hence, $f(p)$
    finds its maximum at $p^*=0$, and is strictly decreasing for $p>0$.\vspace{0.2cm}\newline
2) Consider now $\epsilon>0$. The sign of \eqref{eq: appendix} is analyzed by
focusing on its numerator only, which is rewritten for clarity as:
\begin{equation}
n(p)=r'(p) \epsilon + \left[r'(p)p-r(p) \right]
\end{equation}
We analyze $n(p)$ term by term:
    \begin{itemize}
    \item The first term $r'(p) \epsilon$ is a positive and strictly decreasing function, due to the increasing and strictly concave property of $r(p)$, respectively.
    \item The term in between brackets $r'(p)p-r(p)$ is equal to zero if $p=0$, and a strictly negative (as shown above), strictly decreasing function for $p>0$.  Indeed, the strict
    concavity of $r(p)$  guarantees that the derivative of this term $r''(p)p$ is strictly negative for all $p>0$.
    \end{itemize}
    Consequently, overall $n(p)$ is a strictly decreasing function of $p$ for $p\geq0$. More precisely, the
    lower $\epsilon$, the more rapid the decrease of $n(p)$ will be.
    The initial value at $p=0$ is positive and proportional to $\epsilon$:
    $n(0)=r'(0)\epsilon\geq 0$. On the other hand, the asymptotic value of $n(p)$ as $p\rightarrow \infty$ is
    negative:
    $\lim_{p\rightarrow\infty}n(p)=\lim_{p\rightarrow\infty}\left[r'(p)p-r(p) \right]<0$,where the inequality follows from the strict concavity of $r(p)$ together with the fact that
    $r(0)=0$. Between these two extremes, the strict decrease of $n(p)$ guarantees that $f'(p)$ changes sign  only once (from
positive to negative) at some finite value denoted by $p^*$, and that it will
remain strictly negative for all $p > p^*$. Hence, we have that:
       \begin{itemize}
    \item[(i)] $f(p)$ has a unique maximum, which is achieved at some finite value of \mbox{$p\geq 0$}, denoted by
    $p^*$. The lower the value of $\epsilon$, the lower the value of $p^*$ will
    be.
    \item[(ii)] $f(p)$ is strictly decreasing for $p > p^*$.
    \end{itemize}

\subsection{Proof of Optimality of Algorithm \ref{algo NT problem}}
\label{Appendix proof algo}

If the inequalities in \eqref{eq:ineq algo} hold with $k=N$ (as in
Fig.~\ref{f:NP problem finite}),
the optimal solution provided in Theorem~\ref{Theorem line of sight problem} is
produced by Algorithm~\ref{algo NT problem} in \eqref{eq:power algo} with
$k=N$.

Consider now that the inequalities in \eqref{eq:ineq algo} do not hold for
$k=N$. Then, denote by $k$ the highest $k<N$ for which \eqref{eq:ineq algo}
holds. This situation is depicted in Fig.~\ref{f:NP problem finite no line of
sight}. We first argue that the optimal solution is such that it empties the
battery before receiving the $(k+1)^{th}$ energy packet, i.e. before $t_{k+1}$.
Put differently, the optimal transmitted energy curve should intersect $U(t)$
at a time $t\leq t_{k+1}$. Assume that the opposite holds, as depicted in
Fig.~\ref{f:NP problem finite no line of sight proof}. Then, at some time
$t'\geq t_{k+1}$, the slope of the transmitted energy curve $E(t)$ has to
increase in order to guarantee to empty the battery at time $t=T$ (which is a
necessary condition for optimality). However, it is easy to realize (see the
dot-dashed curve in Fig.~\ref{f:NP problem finite no line of sight proof}) that
such strategy is suboptimal since it violates Theorem~\ref{t:unique}. Note that
the feasibility of the dot-dashed curve in Fig.~\ref{f:NP problem finite no
line of sight proof} in ensured by considering the largest $k$ rather than any
$k$ satisfying \eqref{eq:ineq algo}. Now, since the battery has to be emptied
before receiving the $(k+1)^{th}$ energy packet, we can optimally decouple the
problem. First, the $k$ packet problem with deadline $t_{k+1}=\sum_{n=1}^k T_n$
is solved independently. This subproblem satisfies the inequalities in
\eqref{eq:inequalities suff condition}, such that Theorem~\ref{Theorem line of
sight problem} guarantees that its optimal solution is obtained by
Algorithm~\ref{algo NT problem} in \eqref{eq:power algo}. Then, proceeding
recursively, the algorithm is run for the remaining $N-k$ packet problem which
can be considered as a new problem with an empty battery at the origin.

\begin{figure}
\centering \small \psfrag{time}{time}\psfrag{e}{energy}
\psfrag{Uti}{{\color{red}$\widetilde{U}(t)$}}\psfrag{Ht}{$H(t)$}
\psfrag{Et}{{\color{blue}$E(t)$}} \psfrag{p}{{\color{blue}$p$}}
\psfrag{E1}{$E_1$} \psfrag{eps}{{\color{red}-$\epsilon$}}\psfrag{T}{$T$}
\psfrag{s0}{{\color{blue}$s_0$}}\psfrag{s1}{{\color{blue}$s_1$}}\psfrag{sn}{{\color{blue}$s_n$}}
\psfrag{A1}{$A_1$}
\psfrag{A2}{$A_2$}\psfrag{A3}{$A_k$}\psfrag{AN}{$A_N$}\psfrag{An}{$A_n$}\psfrag{E1}{$E_1$}
\psfrag{E2}{$E_2$} \psfrag{E3}{$E_3$} \psfrag{a1}{$t_1$} \psfrag{a2}{$t_2$}
\psfrag{a3}{$t_3$} \psfrag{a4}{$t_{k+1}$} \psfrag{A}{$A$}
\includegraphics[width=2.7in]{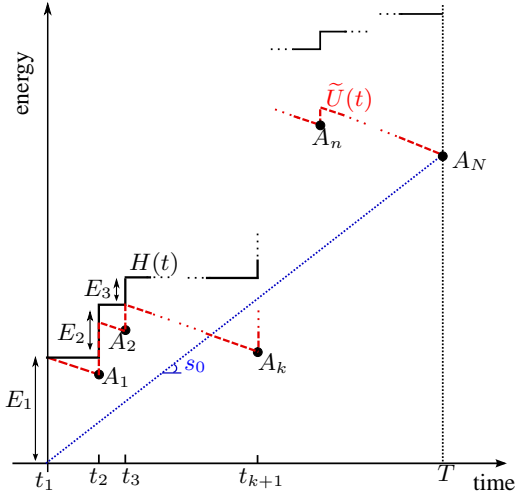}
\caption{General $N_T$ problem.} \label{f:NP problem finite no line of sight}
\end{figure}

\begin{figure}
\centering \small \psfrag{time}{time}\psfrag{e}{energy}
\psfrag{Uti}{{\color{red}$\widetilde{U}(t)$}}\psfrag{Ht}{$H(t)$}
\psfrag{Et}{{\color{blue}$E(t)$}} \psfrag{p}{{\color{blue}$p$}}
\psfrag{E1}{$E_1$} \psfrag{eps}{{\color{red}-$\epsilon$}}\psfrag{T}{$T$}
\psfrag{s0}{{\color{blue}$s_0$}}\psfrag{s1}{{\color{blue}$s_1$}}\psfrag{sn}{{\color{blue}$s_n$}}
\psfrag{A1}{$A_1$}
\psfrag{A2}{$A_2$}\psfrag{A3}{$A_k$}\psfrag{AN}{$A_N$}\psfrag{An}{$A_n$}\psfrag{E1}{$E_1$}
\psfrag{E2}{$E_2$} \psfrag{E3}{$E_3$} \psfrag{a1}{$t_1$} \psfrag{a2}{$t_2$}
\psfrag{a3}{$t_3$} \psfrag{a4}{$t_{k+1}$} \psfrag{A}{$A$}
\psfrag{t}{\small{$\sum_{n=1}^k T_n$}} \psfrag{tt}{$t'$}
\includegraphics[width=2.7in]{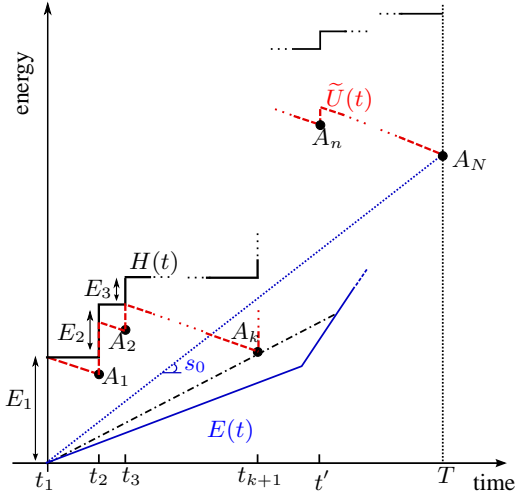}
\caption{General $N_T$ problem: suboptimal solution.} \label{f:NP problem
finite no line of sight proof}
\end{figure}

\end{appendix}

\bibliographystyle{IEEEtran}
\bibliography{ref}

\end{document}